\documentclass[11pt]{amsart}
\def\isdraft{0}

\usepackage{txfonts}
\usepackage{bbm}
\usepackage{graphicx}
\usepackage{bbm}
\usepackage{amsaddr} 
\usepackage{mathtools}
\usepackage{amsmath,amssymb,amsfonts,dsfont}
\usepackage{prettyref} 
\usepackage[utf8]{inputenc}
\usepackage{enumitem}
\usepackage[hyphens]{url}
\usepackage{tikz-cd}
\usepackage{tikz}
\usetikzlibrary{positioning}
\usetikzlibrary{arrows}
\usepackage{bussproofs}
\usepackage[mathscr]{euscript}
\usepackage{hyperref} 
\usepackage{amsthm}
\usepackage{theapa}
\usepackage{a4wide}
\usepackage[color=black,textcolor=white\if\isdraft0,disable\fi]{todonotes}
\usepackage{stackengine}
\usepackage[normalem]{ulem}
\usepackage{stmaryrd}
\usepackage{wrapfig}
\usepackage{pdfpages}
\usepackage{float}
\graphicspath{{fig/}}
\usetikzlibrary{arrows,automata,topaths,matrix,positioning,fit}
\usepgflibrary{shapes.geometric}
\usetikzlibrary{shapes.geometric}
\tikzset{every state/.style={minimum size=0pt}}

\newtheorem{theorem}{Theorem}

\newtheorem{corollary}[theorem]{Corollary}
\newtheorem{fact}[theorem]{Fact}
\newtheorem{lemma}[theorem]{Lemma}
\newtheorem{proposition}[theorem]{Proposition}

\theoremstyle{definition} 

\newtheorem{definition}[theorem]{Definition}

\newtheorem{example}[theorem]{Example}

\newtheorem{notation}[theorem]{Notation}

\newtheorem{problem}[theorem]{Problem}

\newtheorem{remark}[theorem]{Remark}

\newrefformat{alg}{Algorithm \ref{#1}}
\newrefformat{analogy}{Analogy \ref{#1}}
\newrefformat{answer}{Answer \ref{#1}}
\newrefformat{app}{Appendix \ref{#1}}
\newrefformat{cha}{§\ref{#1}}
\newrefformat{claim}{Claim \ref{#1}}
\newrefformat{con}{Convention \ref{#1}}
\newrefformat{conjecture}{Conjecture \ref{#1}}
\newrefformat{cor}{Corollary \ref{#1}}
\newrefformat{equ}{(\ref{#1})}
\newrefformat{err}{Error \ref{#1}}
\newrefformat{exa}{Example \ref{#1}}
\newrefformat{exe}{Exercise \ref{#1}}
\newrefformat{def}{Definition \ref{#1}}
\newrefformat{done}{Done \ref{#1}}
\newrefformat{fact}{Fact \ref{#1}}
\newrefformat{fazit}{Fazit \ref{#1}}
\newrefformat{fig}{Figure \ref{#1}}
\newrefformat{idea}{Idea \ref{#1}}
\newrefformat{ite}{Item \ref{#1}}
\newrefformat{lem}{Lemma \ref{#1}}
\newrefformat{note}{Note \ref{#1}}
\newrefformat{not}{Notation \ref{#1}}
\newrefformat{obs}{Observation \ref{#1}}
\newrefformat{problem}{Problem \ref{#1}}
\newrefformat{prop}{Proposition \ref{#1}}
\newrefformat{pseudocode}{Pseudocode \ref{#1}}
\newrefformat{Q}{Q. \ref{#1}}
\newrefformat{que}{Question \ref{#1}}
\newrefformat{quiz}{Quiz \ref{#1}}
\newrefformat{rem}{Remark \ref{#1}}
\newrefformat{sec}{§\ref{#1}}
\newrefformat{tab}{Table \ref{#1}}
\newrefformat{thm}{Theorem \ref{#1}}
\newrefformat{warning}{Warning \ref{#1}}

\title{Sequential composition of propositional logic programs}
\author{
    Christian Anti\'c
}
\address{
    christian.antic@icloud.com\\
    Vienna, Austria
}

\begin{document}
\begin{abstract} 
    This paper introduces and studies the sequential composition and decomposition of propositional logic programs. We show that acyclic programs can be decomposed into single-rule programs and provide a general decomposition result for arbitrary programs. We show that the immediate consequence operator of a program can be represented via composition which allows us to compute its least model without any explicit reference to operators. This bridges the conceptual gap between the syntax and semantics of a propositional logic program in a mathematically satisfactory way.
\end{abstract}
\maketitle

\section{Introduction}

Rule-based reasoning is an essential part of human intelligence prominently formalized in artificial intelligence research via logic programs \cite<cf.>{Apt90,Apt97,Hodges94,Lloyd87,Makowsky87,Sterling94}. Logic programming is a well-established subfield of theoretical computer science and AI with applications to such diverse fields as expert systems, database theory, diagnosis, planning, learning, natural language processing, and many others \cite<cf.>{Baral03,Coelho88,Kowalski79,Pereira02}. 

The propositional Horn fragment studied in this paper is particularly relevant in database theory via the query language datalog \cite<cf.>{Ceri90} and in answer set programming \cite{Brewka11}, a prominent and successful dialect of logic programming incorporating negation as failure \cite{Clark78} and many other constructs such as aggregates \cite<cf.>{Faber04,Pelov04}, external sources \cite{Eiter05}, and description logic atoms \cite{Eiter08a}.

Describing complex programs as the composition of simpler ones is a common strategy in computer programming and logic programming is no exception as is witnessed by the large amount of research on \textbf{modular logic programming} in the 1980s and 1990s \shortcite<e.g.>{Bossi96,Brogi92,Brogi92a,Brogi95,Brogi99,Bugliesi94,Dong90,Gaifman89,Hill94,Mancarella88,OKeefe85}, and more recently on modular answer set programming \shortcite<e.g.>{Oikarinen06,Oikarinen06a,DaoTran09}.

The \textbf{purpose of this paper} is to add to the logic programmer's repertoire another modularity operation in the form of the sequential composition of propositional logic programs, which is naturally induced by their rule-like structure. The close relationship to the modularity operations introduced in \citeA{OKeefe85} and \citeA{Bugliesi94} is discussed in detail in \prettyref{sec:RW}. 

The \textbf{motivation} for introducing a novel operation is because we feel the need for a \textit{syntactic} composition operation complementing the semantic ones from the literature\footnote{This is partly motivated by recent work on algebraic logic program synthesis via analogical proportions in \citeA{Antic23-23}.} --- in \prettyref{sec:RW} we show that the semantic composition operators can be recovered from the syntactic one (but not vice versa which is evident from their definition). Thus, in a sense, syntactic composition is more fundamental than its semantic counterpart.

The \textbf{main results} of this paper can be summarized as follows:
\begin{enumerate}
    \item We introduce the sequential composition of propositional logic programs and show that the space of all propositional logic programs over some fixed alphabet is closed under composition which means that composition is a total function defined for all pairs of programs. Composition is in general non-associative (\prettyref{exa:non-associativity}). The composition of propositional Krom programs consisting only of rules with at most one body atom is associative and gives rise to the algebraic structure of a monoid, and the restricted class of proper propositional Krom programs consisting only of rules with exactly one body atom yields an idempotent semiring (\prettyref{thm:Krom}).

    \item In \prettyref{sec:Decomposition}, we study {\em decompositions} of programs. The main results here are that acyclic programs \cite<cf.>{Apt91} can be decomposed into a product of single-rule programs (\prettyref{thm:acyclic}), followed by some general results on decompositions of programs in \prettyref{sec:General_Decompositions}.

    \item On the semantic side, the van Emden-Kowalski immediate consequence operator $T_P$ is at the core of logic programming and we show that it can be represented via composition (\prettyref{thm:T_P}), which allows us to compute its least model semantics without any explicit reference to operators (\prettyref{thm:LM}). This bridges the conceptual gap between the syntax and semantics of a propositional logic program in a mathematically satisfactory way.

\end{enumerate}

In a broader sense, this paper is a further step towards an algebra of logic programs and in the future we plan to adapt and generalize the methods of this paper to wider classes of programs, most importantly to first-, and higher-order logic programs \cite{Apt90,Chen93,Lloyd87,Miller12}, and non-monotonic logic programs under the stable model \cite{Gelfond91} or answer set semantics and extensions thereof \cite<cf.>{Baral03,Brewka11,Eiter09,Lifschitz19} --- see \citeA{Antic21-2}.


\section{Preliminaries}

In this section, we first recall the algebraic structures occurring in the rest of the paper, and then we recall the syntax and semantics of propositional logic programs.

\subsection{Algebraic Structures}\label{sec:Algebraic_Structures}

We recall some basic algebraic notions and notations by mainly following the lines of \citeA{Howie03}. 

Given two sets $A$ and $B$, we write $A\subseteq_k B$ in case $A$ is a subset of $B$ with $k$ elements, for some non-negative integer $k$.
We denote the {\em identity function} on a set $A$ by $Id_A$. A {\em permutation} of a set $A$ is any mapping $A\to A$ which is one-to-one and onto. We denote the {\em composition} of two functions $f:A\to A$ and $g:A\to A$ by $f\circ g$ with the usual definition $(f\circ g)(x)=f(g(x))$.

We call a binary relation $\leq$ {\em reflexive} if $x\leq x$, {\em anti-symmetric} if $x\leq y$ and $y\leq x$ implies $x=y$, and {\em transitive} if $x\leq y$ and $y\leq z$ implies $x\leq z$, for all $x,y,z$. A {\em partially ordered set} (or {\em poset}) is a set $L$ together with a reflexive, transitive, and anti-symmetric binary relation $\leq$ on $L$. A {\em prefixed point} of an operator $f$ on a poset $L$ is any element $x\in L$ such that $f(x)\leq x$; moreover, we call any $x\in L$ a {\em fixed point} of $f$ if $f(x)=x$.

A {\em magma} is a set $M$ together with a binary operation $\cdot$ on $M$. We call $(M,\cdot,1)$ a {\em unital} magma if it contains a unit element 1 such that $1x=x1=x$ holds for all $x\in M$. A {\em semigroup} is a magma $(S,\cdot)$ in which $\cdot$ is associative. A {\em monoid} is a semigroup containing a unit element 1 such that $1x=x1=x$ holds for all $x$. A {\em group} is a monoid which contains an inverse $x^{-1}$ for every $x$ such that $xx^{-1}=x^{-1}x=1$. A {\em left} (resp., {\em right}) {\em zero} is an element $0$ such that $0x=0$ (resp., $x0=0$) holds for all $x\in S$. An {\em ordered} semigroup is a semigroup $S$ together with a partial order $\leq$ that is compatible with the semigroup operation, meaning that $x\leq y$ implies $zx\leq zy$ and $xz\leq yz$ for all $x,y,z\in S$. An {\em ordered} monoid is defined in the obvious way. 
A non-empty subset $I$ of $S$ is called a {\em left} (resp., {\em right}) {\em ideal} if $SI\subseteq I$ (resp., $IS\subseteq I$), and a ({\em two-sided}) {\em ideal} if it is both a left and right ideal. An element $x\in S$ is {\em idempotent} if $x\cdot x=x$.

A {\em seminearring} is a set $S$ together with two binary operations $+$ and $\cdot$ on $S$, and a constant $0\in S$, such that $(S,+,0)$ is a monoid and $(S,\cdot)$ is a semigroup satisfying the following laws:
\begin{enumerate}
    \item $(x+y)\cdot z=x\cdot z+y\cdot z$ for all $x,y,z\in S$ (right-distributivity); and
    \item $0\cdot x=0$ for all $x\in S$ (left zero).
\end{enumerate} We say that $S$ is {\em idempotent} if $x+x=x$ holds for all $x\in S$. Typical examples of seminearrings are given by the set of all mappings on a monoid together with composition of mappings, point-wise addition of mappings, and the zero function.

A {\em semiring} is a seminearring $(S,+,\cdot,0)$ such that $+$ is commutative and additionally to the laws of a seminearring the following laws are satisfied:
\begin{enumerate}
    \item $x\cdot (y+z)=x\cdot y+x\cdot z$ for all $x,y,z\in S$ (left-distributivity); and
    \item $x\cdot 0=0$ for all $x\in S$ (right zero).
\end{enumerate} For example, $(\mathbb N,+,\cdot,0)$ and $(2^A,\cup,\cap,\emptyset)$ are semirings.

\subsection{Propositional Logic Programs}\label{sec:Horn}

We recall the syntax and semantics of propositional logic programs.

\subsubsection{Syntax}\label{sec:Syntax}

In the rest of the paper, $A$ denotes a finite alphabet of propositional atoms.

A ({\em propositional Horn logic}) {\em program} over $A$ is a finite set of {\em rules} of the form
\begin{align}\label{equ:r} 
    a_0\leftarrow a_1,\ldots,a_k,\quad k\geq 0,
\end{align} where $a_0,\ldots,a_k\in A$ are propositional atoms. It will be convenient to define, for a rule $r$ of the form \prettyref{equ:r},
\begin{align*} 
     h(r):=\{a_0\} \quad\text{and}\quad  b(r):=\{a_1,\ldots,a_k\}
\end{align*} extended to programs by
\begin{align*} 
     h(P):=\bigcup_{r\in P} h(r) \quad\text{and}\quad  b(P):=\bigcup_{r\in P} b(r).
\end{align*} In this case, the {\em size} of $r$ is $k$ and denoted by $sz(r)$. 

A {\em fact} is a rule with empty body and a {\em proper rule} is a rule which is not a fact. We denote the facts and proper rules in $P$ by $ f(P)$ and $ p(P)$, respectively. 

We call a rule $r$ of the form \prettyref{equ:r} {\em Krom}\footnote{Krom rules were first introduced and studied by \cite{Krom67} and are sometimes called ``binary'' in the literature; we prefer ``Krom'' here since we reserve ``binary'' for programs with two body atoms.} if it has at most one body atom, and we call $r$ {\em binary} if it contains at most two body atoms. A {\em tautology} is any Krom rule of the form $a\leftarrow a$, for $a\in A$. We call a program {\em Krom} if it contains only Krom rules, and we call it {\em binary} if it consists only of binary rules. 

A program is called a {\em single-rule program} if it contains exactly one non-tautological rule. 

We call a program {\em minimalist} if it contains at most one rule for each rule head.


Given two programs $P$ and $R$, we say that $P$ {\em depends on} $R$ if the intersection of $ b(P)$ and $ h(R)$ is not empty. 

We call $P$ {\em acyclic} if there is a mapping $\ell:A\to\{0,1,2,\ldots\}$ such that for each rule $r\in P$, we have $\ell( h(r))>\ell( b(r))$, and in this case we call $\ell$ a {\em level mapping} for $P$. Of course, every level mapping $\ell$ induces an ordering on rules via $r\leq_\ell s$ if $\ell( h(r))\leq\ell( h(s))$. We can transform this ordering into a total ordering by arbitrarily choosing a particular linear ordering within each level, that is, if $\ell(a)=\ell(b)$ then we can choose between $a<_\ell b$ or $b<_\ell a$.

\begin{example}\label{exa:exa} The program
\begin{align*} 
    P:= \left\{
    \begin{array}{l}
        a\\
        b\leftarrow a\\
        c\leftarrow a,b            
    \end{array}
    \right\}
\end{align*} is acyclic with level mapping $\ell(a)=0$, $\ell(b)=1$, $\ell(c)=2$. Adding the rule $a\leftarrow c$ to $P$ yields a non-acyclic program since $\ell(a)\not >\ell(c)$.
\end{example}

Define the {\em dual} of $P$ by
\begin{align*} 
    P^  d:= f(P)\cup\{b\leftarrow  h(r)\mid r\in  p(P),b\in  b(r)\}.
\end{align*} Roughly, we obtain the dual of a program by reversing all the arrows of its proper rules.

\subsubsection{Semantics}

An {\em interpretation} is any set of atoms from $A$. We define the {\em entailment relation}, for every interpretation $I$, inductively as follows: (i) for an atom $a$, $I\models a$ if $a\in I$; (ii) for a set of atoms $B$, $I\models B$ if $B\subseteq I$; (iii) for a rule $r$ of the form \prettyref{equ:r}, $I\models r$ if $I\models  b(r)$ implies $I\models  h(r)$; and, finally, (iv) for a propositional logic program $P$, $I\models P$ if $I\models r$ holds for each rule $r\in P$. In case $I\models P$, we call $I$ a {\em model} of $P$. The set of all models of $P$ has a least element with respect to set inclusion called the {\em least model} of $P$ and denoted by $LM(P)$. We say that $P$ and $R$ are {\em equivalent} if $LM(P)=LM(R)$.

Define the {\em van Emden-Kowalski operator} of $P$, for every interpretation $I$, by
\begin{align*} 
    T_P(I):=\{ h(r)\mid r\in P:I\models  b(r)\}.
\end{align*} The van Emden-Kowalski operator is at the core of logic programming since we have the following well-known operational characterization of models \cite{vanEmden76}:

\begin{proposition}\label{prop:prefixed} An interpretation $I$ is a model of $P$ iff $I$ is a prefixed point of $T_P$.
\end{proposition}

We call an interpretation $I$ a {\em supported model} of $P$ if $I$ is a fixed point of $T_P$. We say that $P$ and $R$ are {\em subsumption equivalent} \cite{Maher88} if $T_P=T_R$, denoted by $P\equiv_{ss}R$.

The {\em least fixed point} computation of $T_P$ is defined by
\begin{align*} 
    &T_P^0=\emptyset,\\
    &T_P^{n+1}=T_P(T_P^n),\\
    &T_P^\infty=\bigcup_{n\geq 0}T_P^n.
\end{align*} The following constructive characterization of least models is due to \cite{vanEmden76}.

\begin{proposition}\label{prop:LM} The least model of a propositional logic program coincides with the least fixed point of its associated van Emden-Kowalski operator, that is, for any program $P$ we have
\begin{align} 
    LM(P)= T_P^\infty.
\end{align}
\end{proposition}

\begin{example} Reconsider the program $P$ of \prettyref{exa:exa}. Intuitively, we expect the least model of $P$ to include all the atoms $a,b,$ and $c$. The following least fixed point computation shows that this is indeed the case:
\begin{align*} 
    &T_P(\emptyset)=\{a\}\\
    &T_P(\{a\})=\{a,b\}\\
    &T_P(\{a,b\})=\{a,b,c\}\\
    &T_P(\{a,b,c\})=\{a,b,c\}.
\end{align*}
\end{example}

\section{Sequential Composition}\label{sec:Sequential_Composition}

We are interested in the algebraic structure of the space of all propositional logic programs. For this we define the sequential composition operation on programs and show in \prettyref{thm:Horn} that the so-obtained space is closed under such compositions. This will allow us later to decompose programs into simpler ones (\prettyref{sec:Decomposition}), to represent the immediate consequence operator on syntactic level (\prettyref{sec:T_P}), and to define an algebraic semantics of programs without any explicit reference to operators (\prettyref{sec:LM}).

\begin{notation} In the rest of the paper, $P$ and $R$ denote propositional logic programs over some joint alphabet $A$.
\end{notation}

We are ready to introduce the main notion of the paper.

\begin{definition}\label{def:PR} We define the ({\em sequential}) {\em composition} of $P$ and $R$ by
\begin{align*} 
    P\circ R=\left\{ h(r)\leftarrow  b(S)\;\middle|\; r\in P, S\subseteq_{sz(r)} R, h(S)= b(r)\right\}.
\end{align*} We will write $PR$ in case the composition operation is understood.
\end{definition}

Roughly, we obtain the composition of $P$ and $R$ by resolving all body atoms in $P$ with `matching' rule heads of $R$. The following example shows why we use $\subseteq_{sz(r)}$ instead of $\subseteq$ in the above definition.

\begin{example} Let $r$ and $R$ be given by
\begin{align*} 
    r:=a\leftarrow b \quad\text{and}\quad R:= \left\{
    \begin{array}{l}
        b\leftarrow c\\
        b\leftarrow d             
    \end{array}
    \right\}.
\end{align*} Since $r$ consists of a single body atom, we have $sz(r)=1$. This means that $\{r\}\circ R$ consists of all rules $a\leftarrow  b(S)$ so that $S\subseteq_1 R$ (i.e., $S$ consists of a single rule) and $ h(S)=\{b\}$. This yields
\begin{align*} 
    \{r\}\circ R= \left\{
    \begin{array}{l}
        a\leftarrow c\\
        a\leftarrow d             
    \end{array}
    \right\}.
\end{align*} This coincides with what we intuitively expect from the composition of $r$ and $R$. On the other hand, if we define composition using $\subseteq$ instead of $\subseteq_{sz(r)}$, then the above requirement changes to $S\subseteq R$ such that $ h(S)=\{b\}$, which means that we can now put $S=R$ adding the undesired rule $a\leftarrow  b(S)=a\leftarrow c,d$ to $\{r\}\circ R$.
\end{example}

Notice that we can reformulate sequential composition as
\begin{align}\label{equ:bigcup} 
    P\circ R=\bigcup_{r\in P}(\{r\}\circ R),
\end{align} which directly implies right-distributivity of composition, that is,
\begin{align}\label{equ:(P_cup_Q)_circ_R} 
    (P\cup Q)\circ R=(P\circ R)\cup (Q\circ R)\quad\text{holds for all propositional logic programs }P,Q,R.
\end{align}

\begin{example} The following counter-example shows that left-distributivity fails in general:  
\begin{align*} 
    \{a\leftarrow b,c\}\circ(\{b\}\cup\{c\})=\{a\} \quad\text{whereas}\quad (\{a\leftarrow b,c\}\circ\{b\})\cup(\{a\leftarrow b,c\}\circ\{c\})=\emptyset.
\end{align*} 
\end{example}

We can write $P$ as the union of its facts and proper rules, that is,
\begin{align}\label{equ:P} 
    P= f(P)\cup  p(P).
\end{align} Hence, we can rewrite the composition of $P$ and $R$ as
\begin{align}\label{equ:facts_proper} 
    PR&=( f(P)\cup  p(P))R\stackrel{\prettyref{equ:(P_cup_Q)_circ_R}}= f(P)R\cup  p(P)R
    = f(P)\cup  p(P)R,
    \end{align} which shows that the facts in $P$ are preserved by composition, that is, we have
\begin{align}\label{equ:facts_PR} 
     f(P)\subseteq  f(PR).
\end{align} The facts of the composition of to two programs is given by
\begin{align*} 
     f(PR)\stackrel{\prettyref{equ:facts_proper}}= f( f(P)\cup  p(P)R)= f( f(P))\cup  f( p(P)R)= f(P)\cup  p(P) f(R).
\end{align*} We can compute the heads and bodies via
\begin{align}\label{equ:head_P} 
     h(P)=PA \quad\text{and}\quad  b(P)= p(P)^  d A.
\end{align} Moreover, we have
\begin{align*} 
     h(PR)\subseteq  h(P) \quad\text{and}\quad  b(PR)\subseteq  b(R).
\end{align*}

Define the {\em unit program} (over $A$) by the propositional Krom program
\begin{align*} 
    1_A:=\{a\leftarrow a\mid a\in A\}.
\end{align*} In the sequel, we will often omit the reference to $A$.

We are now ready to state the main structural result of the paper:

\begin{theorem}\label{thm:Horn} The space of all propositional logic programs over some fixed alphabet forms a finite unital magma with respect to composition ordered by set inclusion with the neutral element given by the unit program. Moreover, the empty program is a left zero and composition distributes from the right over union, that is, for any programs $P,Q,R$ we have
\begin{align} 
    P1&=1P=P\\
    \label{equ:0P=0}\emptyset P&=\emptyset\\
    \label{equ:P_cup_R_circ_R} (P\cup R)Q&=(PQ)\cup (RQ).
\end{align}
\end{theorem}
\begin{proof} The space of all propositional logic programs is obviously closed under composition, which shows that it forms a magma.

We proceed by proving that 1 is neutral with respect to composition. By definition of composition, we have
\begin{align*} 
    P1&=\{ h(r)\leftarrow  b(S)\mid r\in P,S\subseteq_{sz(r)} 1, h(S)= b(r)\}.
\end{align*} Now, by definition of 1 and $S\subseteq_{sz(r)} 1$, we have $ h(S)= b(S)$ and therefore $ b(S)= b(r)$. Hence,
\begin{align*} 
    P1=P.
\end{align*} Similarly, we have
\begin{align*} 
    1P&=\{ h(r)\leftarrow  b(S)\mid r\in 1,S\subseteq_1 P, h(S)= b(r)\}.
\end{align*} As $S$ is a subset of $P$ with a single element, $S$ is a singleton $S=\{s\}$, for some rule $s$ with $ h(s)= b(r)$ and, since $ b(r)= h(r)$ holds for every rule in 1, $ h(s)= h(r)$. Hence,
\begin{align*} 
    1P&=\{ h(r)\leftarrow  b(s)\mid r\in 1, s\in P, h(s)= b(r)\}\\
    &=\{ h(s)\leftarrow  b(s)\mid s\in P\}\\
    &=P.
\end{align*} This shows that 1 is neutral with respect to composition. That composition is ordered by set inclusion is obvious. We now turn our attention to the operation of union. In \prettyref{equ:(P_cup_Q)_circ_R} we argued for the right-distributivity of composition. That the empty set is a left zero is obvious. As union is idempotent, the set of all propositional logic programs forms the structure of a finite idempotent seminearring.
\end{proof}

The following example shows that, unfortunately, composition is {\em not} associative.\footnote{In \prettyref{thm:PMN} we will see that composition is associative for minimalist programs and in \prettyref{cor:assoc_ss} we will see that it is associative modulo subsumption equivalence.}

\begin{example}\label{exa:non-associativity} Consider the rule
\begin{align*} 
    r:=a\leftarrow b,c
\end{align*} and the programs
\begin{align*} 
    P:= \left\{
    \begin{array}{l}
        b\leftarrow b\\
        c\leftarrow b,c
    \end{array}
    \right\} \quad\text{and}\quad R:= \left\{
    \begin{array}{l}
        b\leftarrow d\\
        b\leftarrow e\\
        c\leftarrow f
    \end{array}
    \right\}.
\end{align*} Let us compute $(\{r\}P)R$. We first compute $\{r\}P$ by noting that the rule $r$ has two body atoms and has therefore size 2, which means that there is only a single choice of subprogram $S$ of $P$ with two rules, namely $S=P$; this yields
\begin{align*} 
    \{r\}P=\{ h(r)\leftarrow  b(P)\}=\{a\leftarrow b,c\}=\{r\}.
\end{align*} Next we compute $(\{r\}P)R=\{r\}R$ by noting that now there are two possible $S_1,S_2\subseteq_2 R$ with $ h(S_1)= h(S_2)= b(r)$, given by
\begin{align*} 
    S_1= \left\{
    \begin{array}{l}
        b\leftarrow d\\
        c\leftarrow f             
    \end{array}
    \right\} \quad\text{and}\quad S_2= \left\{
    \begin{array}{l}
        b\leftarrow e\\
        c\leftarrow f             
    \end{array}
    \right\}.
\end{align*} This yields
\begin{align*} 
    \{r\}R= \left\{
    \begin{array}{l}
         h(r)\leftarrow  b(S_1)\\             
         h(r)\leftarrow  b(S_2)             
    \end{array}
    \right\}= \left\{
    \begin{array}{l}
        a\leftarrow d,f\\             
        a\leftarrow e,f             
    \end{array}
    \right\}.
\end{align*} Let us now compute $\{r\}(PR)$. We first compute $PR$. By \prettyref{equ:bigcup} we have
\begin{align*} 
    PR=\{b\leftarrow b\}R\cup \{c\leftarrow b,c\}R.
\end{align*} We easily obtain
\begin{align*} 
    \{b\leftarrow b\}R= \left\{
    \begin{array}{l}
        b\leftarrow d\\             
        b\leftarrow e             
    \end{array}
    \right\}.
\end{align*} Similar computations as above show
\begin{align*} 
    \{c\leftarrow b,c\}R= \left\{
    \begin{array}{l}
        c\leftarrow d,f\\             
        c\leftarrow d,e             
    \end{array}
    \right\}.
\end{align*} So in total we have
\begin{align*} 
    PR= \left\{
    \begin{array}{l}
        b\leftarrow d\\             
        b\leftarrow e\\
        c\leftarrow d,f\\             
        c\leftarrow d,e             
    \end{array}
    \right\}.
\end{align*} To compute $\{r\}(PR)$, we therefore see that there are four two rule subprograms $S_1,S_2,S_4,S_4\subseteq_2 PR$ with $b$ or $c$ in their heads given by
\begin{align*} 
    S_1= \left\{
    \begin{array}{l}
        b\leftarrow d\\
        c\leftarrow d,f             
    \end{array}
    \right\} \qquad S_2= \left\{
    \begin{array}{l}
        b\leftarrow d\\
        c\leftarrow d,e             
    \end{array}
    \right\} \qquad S_3= \left\{
    \begin{array}{l}
        b\leftarrow e\\
        c\leftarrow d,f             
    \end{array}
    \right\} \qquad S_4= \left\{
    \begin{array}{l}
        b\leftarrow e\\
        c\leftarrow d,f             
    \end{array}
    \right\}.
\end{align*} Hence, we have
\begin{align*} 
    \{r\}PR= \left\{
    \begin{array}{l}
         h(r)\leftarrow  b(S_1)\\             
         h(r)\leftarrow  b(S_2)\\             
         h(r)\leftarrow  b(S_3)\\             
         h(r)\leftarrow  b(S_4)
    \end{array}
    \right\}= \left\{
    \begin{array}{l}
        a\leftarrow d,f\\
        a\leftarrow d,e\\
        a\leftarrow d,e,f             
    \end{array}
    \right\}.
\end{align*} We have thus shown
\begin{align*} 
    \{r\}(PR)= \left\{
    \begin{array}{l}
        a\leftarrow d,f\\
        a\leftarrow e,f\\
        a\leftarrow d,e,f\\
    \end{array}
    \right\}\neq \left\{
    \begin{array}{l}
        a\leftarrow d,f\\
        a\leftarrow e,f
    \end{array}
    \right\}=(\{r\}P)R.
\end{align*}
\end{example}

\section{Restricted Classes of Programs}\label{sec:Restricted}

Now that we have a basic understanding of the sequential composition of programs, we continue by studying composition in some restricted classes of programs like minimalist (\prettyref{sec:Minimalist_Programs}) and Krom programs (\prettyref{sec:Krom_Programs}) where we find that associativity holds. Having a basic understanding of those restricted classes of programs will be useful later when we deal with decompositions of programs where factors of programs often have one of those restricted forms.

\subsection{Minimalist Programs}\label{sec:Minimalist_Programs}

Recall that we call a program {\em minimalist} if it contains at most one rule for each head atom. \prettyref{thm:PMN} shows that minimalist programs satisfy a form of associativity.

Given a minimalist program $M$, in the computation of $P\circ M$ we can omit the reference to $r$ in $\subseteq_{sz(r)}$ in \prettyref{def:PR}, that is, we have
\begin{align*} 
    P\circ M=\{ h(r)\leftarrow  b(S)\mid r\in P,S\subseteq M, h(S)= b(r)\}.
\end{align*}

In \prettyref{exa:non-associativity} we have seen that composition is not associative in general. The situation changes for minimalist programs.

\begin{theorem}\label{thm:PMN} For any programs $P,M,N$ if $M$ and $N$ are minimalist, then
\begin{align}\label{equ:PMN} 
    (PM)N=P(MN).
\end{align}
\end{theorem}
\begin{proof}
As a consequence of \prettyref{equ:bigcup}, a rule $r$ is in $(PM)N$ iff
\begin{align}\label{equ:rr} 
    \{r\}=\left(\{a_0\leftarrow a_1,\ldots,a_k\}\circ \left\{
    \begin{array}{l}
        a_1\leftarrow B_1\\
        \vdots\\
        a_k\leftarrow B_k
    \end{array}
    \right\}\right)\circ \left\{
    \begin{array}{l}
        b_1\leftarrow C_1\\
        \vdots\\
        b_m\leftarrow C_m
    \end{array}
    \right\},
\end{align} for some rules $a_0\leftarrow a_1,\ldots,a_k\in P$, $k\geq 0$, $b_1\leftarrow B_1,\ldots,b_k\leftarrow B_k\in M$, $b_1\leftarrow C_1,\ldots,b_m\leftarrow C_m\in N$, and $B_1\cup\ldots\cup B_k=\{b_1,\ldots,b_m\}$, $m\geq 0$. Since
\begin{align*} 
    \left\{
    \begin{array}{l}
        a_1\leftarrow B_1\\
        \vdots\\
        a_k\leftarrow B_k
    \end{array}
    \right\} \quad\text{and}\quad \left\{
    \begin{array}{l}
        b_1\leftarrow C_1\\
        \vdots\\
        b_m\leftarrow C_m
    \end{array}
    \right\}
\end{align*} are minimalist, we can rewrite \prettyref{equ:rr} as
\begin{align*} 
    \{r\}=\{a_0\leftarrow a_1,\ldots,a_k\}\circ \left(\left\{
    \begin{array}{l}
        a_1\leftarrow B_1\\
        \vdots\\
        a_k\leftarrow B_k
    \end{array}
    \right\}\circ \left\{
    \begin{array}{l}
        b_1\leftarrow C_1\\
        \vdots\\
        b_m\leftarrow C_m
    \end{array}
    \right\}\right),
\end{align*} which shows $r\in P(MN)$. The proof that every rule in $P(MN)$ is in $(PM)N$ is analogous.
\end{proof}

\subsection{Krom Programs}\label{sec:Krom_Programs}

Recall that we call a program {\em Krom} if it contains only rules with at most one body atom. This includes interpretations, unit programs, and permutations. Krom programs often appear as factors in decompositions of programs and it is therefore beneficial to know their basic algebraic properties (\prettyref{thm:Krom}). First notice that for any Krom program $K$ and any program $P$, composition simplifies to
\begin{align*} 
    K\circ P= f(K)\cup\{a\leftarrow B\mid a\leftarrow b\in  p(K)\text{ and }b\leftarrow B\in P\}.
\end{align*}

We have the following structural result as a specialization of \prettyref{thm:Horn}:

\begin{theorem}\label{thm:Krom} The space of all propositional Krom programs forms a monoid with respect to composition with the neutral element given by the unit program, and it forms a seminearring with respect to sequential composition and union. More generally, we have
\begin{align}\label{equ:KPR} 
    K(PR)=(KP)R,
\end{align} for arbitrary programs $P$ and $R$. Moreover, Krom programs distribute from the left, that is, for any programs $P$ and $R$, we have
\begin{align}\label{equ:KP_cup_R} 
    K(P\cup R)=KP\cup KR.
\end{align} This implies that the space of proper propositional Krom programs forms a finite idempotent semiring.
\end{theorem}
\begin{proof} We first prove \prettyref{equ:KPR}, which implies associativity for propositional Krom programs. A rule $r$ is in $K(PR)$ iff either $r=a\in K$, in which case we have $r\in (KP)R$, or
\begin{align*} 
    \{r\}=\{a\leftarrow b\}\circ\left(\{b\leftarrow b_1,\ldots,b_k\}\circ\left\{
    \begin{array}{l}
        b_1\leftarrow B_1\\
        \vdots\\
        b_k\leftarrow B_k
    \end{array}
    \right\}\right),
\end{align*} for some $a\leftarrow b\in K$, $b\leftarrow b_1,\ldots,b_k\in P$, and $b_1\leftarrow B_1,\ldots,b_k\leftarrow B_k\in R$, in which case we have $r=a\leftarrow B_1\cup\ldots\cup B_k$. A simple computation shows
\begin{align*} 
    \{r\}=(\{a\leftarrow b\}\circ\{b\leftarrow b_1,\ldots,b_k\})\circ\left\{
    \begin{array}{l}
      b_1\leftarrow B_1\\
      \vdots\\
      b_k\leftarrow B_k
    \end{array}
    \right\}.
\end{align*} The proof that every rule in $(KP)R$ is in $K(PR)$ is analogous.

We now want to prove \prettyref{equ:KP_cup_R}. For any proper Krom rule $r=a\leftarrow b$, we have
\begin{align*} 
    \{a\leftarrow b\}\circ (P\cup R)&=\{a\leftarrow B\mid b\leftarrow B\in P\cup R, B\subseteq A\}\\
    &=\{a\leftarrow B\mid b\leftarrow B\in P\}\cup\{a\leftarrow B\mid b\leftarrow B\in R\}\\
    &=(\{a\leftarrow b\}\circ P)\cup (\{a\leftarrow b\}\circ R).
\end{align*} Hence, we have
\begin{align*} 
    K(P\cup R)&\stackrel{\prettyref{equ:facts_proper}}= f(K)\cup  p(K)(P\cup R)\\
    &\stackrel{\prettyref{equ:bigcup}}= f(K)\cup\bigcup_{r\in  p(K)}\{r\}(P\cup R)\\
    &= f(K)\cup\bigcup_{r\in  p(K)}(\{r\}P\cup\{r\}R)\\
    &= f(K)\cup\bigcup_{r\in  p(K)}\{r\}P\cup\bigcup_{r\in  p(K)}\{r\}R\\
    &= f(K)\cup  p(K)P\cup  p(K)R\\
    &=( f(K)\cup  p(K)P)\cup ( f(K)\cup  p(K)R)\\
    &=KP\cup KR
\end{align*} where the sixth identity follows from the idempotency of union. This shows that Krom programs distribute from the left which implies that the space of all proper\footnote{If $K$ contains facts, then $K\emptyset= f(K)\neq\emptyset $ violates the last axiom of a semiring (cf. Section \ref{sec:Algebraic_Structures}).} Krom programs forms a semiring (cf. \prettyref{thm:Horn}). 
\end{proof}

\subsubsection{Interpretations}

Notice that, formally, interpretations are Krom programs containing only rules with empty bodies called facts, which gives interpretations a special compositional meaning.

\begin{fact}\label{fact:IP=I} Every interpretation $I$ is a left zero, that is, for any program $P$, we have
\begin{align}\label{equ:IP=I} 
    IP=I.
\end{align} Consequently, the space of interpretations forms a right ideal.\footnote{See \prettyref{cor:I_ideal}.}
\end{fact}

\begin{corollary} A program $P$ commutes with an interpretation $I$ iff $I$ is a supported model of $P$, that is,
\begin{align*} 
    PI=IP \quad\Leftrightarrow\quad I\in Supp(P).
\end{align*}
\end{corollary}
\begin{proof} A direct consequence of \prettyref{fact:IP=I} and the forthcoming \prettyref{thm:T_P} (which is not depending on this result).
\end{proof}

\subsubsection{Permutations}\label{sec:Permutations}

With every permutation $\pi:A\to A$, we associate the propositional Krom program
\begin{align*} 
    \pi=\{\pi(a)\leftarrow a\mid a\in A\}.
\end{align*} We adopt here the standard cycle notation for permutations. For instance, we have
\begin{align*} 
    \pi_{(a\,b)}= \left\{
    \begin{array}{l}
      a\leftarrow b\\
      b\leftarrow a
    \end{array}
    \right\} \quad\text{and}\quad\pi_{(a\,b\,c)}= \left\{
    \begin{array}{l}
      a\leftarrow b\\
      b\leftarrow c\\
      c\leftarrow a
    \end{array}
    \right\}.
\end{align*} The inverse $\pi^{-1}$ of $\pi$ translates into the language of programs as
\begin{align*} 
    \pi^{-1}=\pi^  d.
\end{align*} Interestingly, we can rename the atoms occurring in a program via permutations and composition by
\begin{align*} 
    \pi\circ P\circ \pi^  d=\{\pi( h(r))\leftarrow\pi( b(r))\mid r\in P\}.
\end{align*}

We have the following structural result as a direct instance of the more general result for permutations:

\begin{fact} The space of all permutation programs forms a group.
\end{fact}

\subsection{Reducts}\label{sec:Reducts}

Reducing a program to a restricted alphabet is a fundamental operation on programs which we will often use in the rest of the paper, especially in \prettyref{sec:Decomposition}. For example, in \prettyref{cor:P^h(R)_circ_^b(R)R} we will see that the computation of composition can be simplified using reducts. This motivates the following definition.

\begin{definition} We define the {\em left} and {\em right reduct}\footnote{Our notion of right reduct coincides with the FLP reduct \cite{Faber04} well-known in answer set programming.} of a program $P$, with respect to some interpretation $I$, respectively by
\begin{align*} 
    ^IP:=\{r\in P\mid I\models  h(r)\} \quad\text{and}\quad P^I:=\{r\in P\mid I\models  b(r)\}.
\end{align*}
\end{definition}

Our first observation is that we can compute the facts of $P$ via the right reduct with respect to the empty set, that is, we have
\begin{align}\label{equ:f(P)} 
     f(P)=P^\emptyset =P\emptyset .
\end{align} On the contrary, computing the left reduct with respect to the empty set yields
\begin{align}\label{equ:^0P} 
    ^\emptyset P=\emptyset .
\end{align} Moreover, for any interpretations $I$ and $J$, we have
\begin{align}\label{equ:^JI} 
    ^JI=I\cap J \quad\text{and}\quad I^J=I.
\end{align}

Notice that we obtain the reduction of $P$ to the atoms in $I$, denoted $P|_I$, by
\begin{align}\label{equ:P_I} 
    P|_I={^I}(P^I)=(^IP)^I.
\end{align} As the order of computing left and right reducts is irrelevant, in the sequel we will omit the parentheses in \prettyref{equ:P_I}. Of course, we have
\begin{align*} 
    ^AP=P^A={^A}P^A=P.
\end{align*} Moreover, we have
\begin{align}
    \label{equ:1^I} 1^I&={^I}1={^I}1^I=1_I=1|_I\\
    \label{equ:1I_circ_J}1^I\circ 1^J&=1^{I\cap J}=1^J\circ 1^I=1^I\cap 1^J\\
    \label{equ:1^I_cup_J}1^I\cup 1^J&=1^{I\cup J}.
\end{align}

We now want to relate reducts to composition and union.


\begin{proposition}\label{prop:^IP} For any program $P$ and interpretation $I$, we have
\begin{align}\label{equ:^IP} 
    ^IP=1^I\circ P \quad\text{and}\quad P^I=P\circ 1^I.
\end{align} Moreover, we have
\begin{align}\label{equ:P_I2} 
    P|_I=1^I\circ P\circ 1^I.
\end{align} Consequently, for any program $R$, we have
\begin{align} 
    \label{equ:^IPR}^I(P\cup R)={^I}P\cup{^I}R& \quad\text{and}\quad{^I}(PR)={^I}PR\\
    \label{equ:^IP_cup_R}(P\cup R)^I=P^I\cup R^I& \quad\text{and}\quad (PR)^I=PR^I.
\end{align}
\end{proposition}
\begin{proof} We compute
\begin{align*} 
    1^I\circ P&=\{ h(r)\leftarrow  b(s)\mid r\in 1^I,s\in P, h(s)= b(r)\}\\
    &=\{ h(r)\leftarrow  b(s)\mid r\in 1,s\in P, h(s)= b(r), h(r)= b(r)\in I\}\\
    &=\{ h(s)\leftarrow  b(s)\mid s\in P, h(s)\in I\}\\
    &=\{s\in P\mid I\models  h(s)\}\\
    &={^I}P.
\end{align*} Similarly, we compute
\begin{align}\label{equ:P_circ_1} 
    P\circ 1^I&=\{ h(r)\leftarrow  b(S)\mid r\in P,S\subseteq_{sz(r)} 1^I, h(S)= b(r)\}\\
    &=\left\{ h(r)\leftarrow  b(S) \;\middle|\; 
    \begin{array}{l}
        r\in P,S\subseteq_{sz(r)} 1,\\
         h(S)= b(r),\\
         h(S)= b(S)\subseteq I
    \end{array}\right\}.
\end{align} By definition of $1$, we have $ b(S)= h(S)$ and therefore $ b(S)= b(r)$ and $ b(r)\subseteq I$. Hence, \prettyref{equ:P_circ_1} is equivalent to
\begin{align*} 
    \{ h(r)\leftarrow  b(r)\mid r\in P:I\models  b(r)\}=P^I.
\end{align*}
Finally, we have
\begin{align*} 
    ^I(P\cup R)\stackrel{\prettyref{equ:^IP}}=1^I(P\cup R)\stackrel{\prettyref{equ:KP_cup_R}}=1^IP\cup 1^IR\stackrel{\prettyref{equ:^IP}}={^I}P\cup{^I}R
\end{align*} and
\begin{align*} 
    ^I(PR)\stackrel{\prettyref{equ:^IP}}=(1^I)(PR)\stackrel{\prettyref{equ:KPR}}=((1^I)P)R\stackrel{\prettyref{equ:^IP}}={^IP}R
\end{align*} with the remaining identities in \prettyref{equ:^IP_cup_R} holding by analogous computations.
\end{proof}

Consequently, we can simplify the computation of composition as follows.

\begin{corollary}\label{cor:P^h(R)_circ_^b(R)R} For any programs $P$ and $R$, we have
\begin{align}\label{equ:P^head_circ_^body_R}
    P\circ R=P^{ h(R)}\circ{^{ b(P)}}R.
\end{align}
\end{corollary}

\subsection{Proper Programs}

By equation \prettyref{equ:f(P)}, we can extract the facts of a program $P$ by computing the right reduct with respect to the empty set, $P^\emptyset$, and by composing $P$ with the empty set, $P\emptyset$. Unfortunately, there is no analogous characterization of the proper rules in terms of composition. We have therefore introduced a unary operator $ p$ for that purpose and we now want to state some basic properties of that operator as it is used to characterize idempotent programs (\prettyref{prop:idempotent}). The proper rules operator satisfies the following identities, for any propositional logic programs $P$ and $R$, and interpretation $I$:
\begin{align} 
     p\circ  p&= p\\
    \label{equ:pf}  p\circ  f&= f\circ  p=\emptyset \\
     p(1)&=1\\
     p(I)&=\emptyset \\
     p(P)\emptyset &\stackrel{\prettyref{equ:f(P)}}= f( p(P))\stackrel{\prettyref{equ:pf}}=\emptyset \\
    \label{equ:proper_cup}  p(P\cup R)&= p(P)\cup  p(R).
\end{align} Of course, we have $ p(P)=P$ iff $P$ contains no facts, that is, iff $P^\emptyset =\emptyset $. The last identity says that the proper rules operator is compatible with union; however, the following counter-example shows that it is {\em not} compatible with sequential composition:
\begin{align*} 
     p\left(\{a\leftarrow b,c\}\circ \left\{
    \begin{array}{l}
      b\leftarrow b\\
      c
    \end{array}
    \right\}\right)=\{a\leftarrow b\}
    \end{align*} whereas
    \begin{align*}  p(\{a\leftarrow b,c\})\circ  p\left(\left\{
    \begin{array}{l}
      b\leftarrow b\\
      c
    \end{array}
    \right\}\right)=\emptyset.
\end{align*} This shows that the proper rule operator fails to be a homomorphism with respect to composition.

\begin{proposition} The space of all proper propositional logic programs forms a submagma of the space of all propositional logic programs with the zero given by the empty set.
\end{proposition}
\begin{proof} The space of proper programs is closed under composition. It remains to show that the empty set is a zero, but this follows from the fact that it is a left zero by \prettyref{equ:IP=I} and from the observation that for any proper program $P$, we have $P\emptyset\stackrel{\prettyref{equ:f(P)}}= f(P)=\emptyset .$
\end{proof}

\subsection{Idempotent Programs}

Recall that $P$ is called {\em idempotent} if $P\circ P=P$. Idempotent programs have the nice property that their least models are trivial to compute as they coincide with the facts of the program. Characterizing idempotent programs on an algebraic level is therefore interesting, which is done in the next result in terms of their facts and proper rules.


\begin{proposition}\label{prop:idempotent} A program $P$ is idempotent iff
\begin{align}\label{equ:P^2} 
     p(P) f(P)\subseteq  f(P) \quad\text{and}\quad  p(P^2)= p( p(P)P).
\end{align}
\end{proposition}
\begin{proof} The program $P$ is idempotent iff we have (i) $ f(P^2)= f(P)$ and (ii) $ p(P^2)= p(P)$. For the first condition in \prettyref{equ:P^2}, we compute
\begin{align*} 
     f(P^2)
        &\stackrel{\prettyref{equ:f(P)}}=P^2\emptyset\\
        &=P(P\emptyset)\\ 
        &=P\circ  f(P)\\
        &\stackrel{\prettyref{equ:P}}=( f(P)\cup  p(P)) f(P)\\
        &\stackrel{\prettyref{equ:(P_cup_Q)_circ_R}}= f(P) f(P)\cup  p(P) f(P)\\
        &\stackrel{\prettyref{equ:IP=I}}= f(P)\cup  p(P) f(P).
\end{align*} This yields
\begin{align*} 
     f(P^2)= f(P) \quad\Leftrightarrow\quad  p(P) f(P)\subseteq  f(P).
\end{align*} For the second condition in \prettyref{equ:P^2}, we compute
\begin{align*} 
     p(P^2)
        &= p(( f(P)\cup  p(P))P)\\
        &\stackrel{\prettyref{equ:(P_cup_Q)_circ_R}}= p( f(P)P\cup  p(P)P)\\
        &\stackrel{\prettyref{equ:IP=I}}= p( f(P)\cup  p(P)P)\\
        &\stackrel{\prettyref{equ:proper_cup}}= p( f(P))\cup  p( p(P)P)\\
        &\stackrel{\prettyref{equ:pf}}= p( p(P)P).
\end{align*}
\end{proof}

\begin{fact}\label{fact:I_idempotent} Every interpretation is idempotent.
\end{fact}
\begin{proof} A direct consequence of \prettyref{equ:IP=I}.
\end{proof}





\section{Decomposition}\label{sec:Decomposition}

It is often useful to know how a program is assembled from its simpler parts. In this section, we therefore study sequential decompositions of acyclic (\prettyref{sec:Acyclic_Programs}) and arbitrary programs (\prettyref{sec:General_Decompositions}). Before we do that we first look in \prettyref{sec:Adding_} at some ways how to algebraically transform programs via composition.


\subsection{Adding and Removing Body Atoms}\label{sec:Adding_}

Notice that we can manipulate rule bodies via composition on the right. For example, we have
\begin{align*} 
    \{a\leftarrow b,c\}\circ \left\{
    \begin{array}{l}
      b\leftarrow b\\
      c
    \end{array}
    \right\}=\{a\leftarrow b\}.
\end{align*} The general construction here is that we add a tautological rule $b\leftarrow b$ for every body atom $b$ of $P$ which we want to preserve, and we add a fact $c$ in case we want to remove $c$ from the rule bodies in $P$. We therefore define, for an interpretation $I$, the minimalist program
\begin{align*} 
    I^\ominus=1^{A-I}\cup I.
\end{align*} Notice that $.^\ominus$ is computed with respect to some fixed alphabet $A$. For instance, we have
\begin{align*} 
    A^\ominus=A \quad\text{and}\quad\emptyset ^\ominus=1.
\end{align*} The first equation yields another explanation for \prettyref{equ:head_P}, that is, we can compute the heads in $P$ by removing all body atoms of $P$ via $$ h(P)=PA^\ominus=PA.$$ Interestingly enough, we have
\begin{align*} 
    I^\ominus I=(1^{A-I}\cup I)I\stackrel{\prettyref{equ:(P_cup_Q)_circ_R}}=1^{A-I}I\cup I^2\stackrel{\prettyref{equ:IP=I},\prettyref{equ:^JI},\prettyref{equ:^IP}}=((A-I)\cap I)\cup I=I
\end{align*} and
\begin{align*} 
    I^\ominus P={^{A-I}}P\cup I.
\end{align*} Moreover, in the example above, we have
\begin{align*} 
    \{c\}^\ominus= \left\{
    \begin{array}{l}
      a\leftarrow a\\
      b\leftarrow b\\
      c
    \end{array}
    \right\} \quad\text{and}\quad\{a\leftarrow b,c\}\circ\{c\}^\ominus=\{a\leftarrow b\}
\end{align*} as desired. Notice also that the facts of a program are, of course, not affected by composition on the right, that is, we cannot expect to remove facts via composition on the right (cf. \prettyref{equ:facts_PR}).

We have the following general result:

\begin{proposition}\label{prop:ominus} For any $P$ and interpretation $I$, we have
\begin{align*} 
    PI^\ominus=\{ h(r)\leftarrow ( b(r)-I)\mid r\in P\}.
\end{align*}
\end{proposition}

In analogy to the above construction, we can add body atoms via composition on the right. For example, we have
\begin{align*} 
    \{a\leftarrow b\}\circ\{b\leftarrow b,c\}=\{a\leftarrow b,c\}.
\end{align*} Here, the general construction is as follows. For an interpretation $I$, define the minimalist program
\begin{align*} 
    I^\oplus=\{a\leftarrow(\{a\}\cup I)\mid a\in A\}.
\end{align*} For instance, we have
\begin{align*} 
    A^\oplus=\{a\leftarrow A\mid a\in A\} \quad\text{and}\quad\emptyset ^\oplus=1.
\end{align*} Interestingly enough, we have
\begin{align*} 
    I^\oplus I^\ominus=I^\ominus \quad\text{and}\quad I^\oplus I=I.
\end{align*} Moreover, in the example above, we have
\begin{align*} 
    \{c\}^\oplus= \left\{
    \begin{array}{l}
      a\leftarrow a,c\\
      b\leftarrow b,c\\
      c\leftarrow c
    \end{array}
    \right\} \quad\text{and}\quad\{a\leftarrow b\}\circ\{c\}^\oplus=\{a\leftarrow b,c\}
\end{align*} as desired. As composition on the right does not affect the facts of a program, we cannot expect to append body atoms to facts via composition on the right. However, we can add arbitrary atoms to {\em all} rule bodies simultaneously and in analogy to \prettyref{prop:ominus}, we have the following general result:

\begin{proposition}\label{prop:oplus} For any program $P$ and interpretation $I$, we have
\begin{align*} 
    PI^\oplus= f(P)\cup\{ h(r)\leftarrow ( b(r)\cup I)\mid r\in  p(P)\}.
\end{align*}
\end{proposition}

We now want to illustrate the interplay between the above concepts with an example.

\begin{example} Consider the propositional logic programs
\begin{align*} 
    P= \left\{
    \begin{array}{l}
      c\\
      a\leftarrow b,c\\
      b\leftarrow a,c
    \end{array}
    \right\} \quad\text{and}\quad \pi_{(a\,b)}= \left\{
    \begin{array}{l}
      a\leftarrow b\\
      b\leftarrow a
    \end{array}
    \right\}.
\end{align*} Roughly, we obtain $P$ from $\pi_{(a\,b)}$ by adding the fact $c$ to $\pi_{(a\,b)}$ and to each body rule in $\pi_{(a\,b)}$. Conversely, we obtain $\pi_{(a\,b)}$ from $P$ by removing the fact $c$ from $P$ and by removing the body atom $c$ from each rule in $P$. This can be formalized as\footnote{Here, we have $\{c\}^\ast=1\cup\{c\}$ and $\{c\}^\ast\pi_{(a\,b)}=\pi_{(a\,b)}\cup\{c\}$ by the forthcoming equation \prettyref{equ:I^ast}.}
\begin{align*} 
    P=\{c\}^\ast\pi_{(a\,b)}\{c\}^\oplus \quad\text{and}\quad\pi_{(a\,b)}=1^{\{a,b\}}P\{c\}^\ominus.
\end{align*}
\end{example}

\begin{remark}\label{rem:add_r} It is important to emphasize that we cannot add {\em arbitrary} rules to a program via composition. For example, suppose we want to add the rule $r=a\leftarrow b$ to the empty program $\emptyset$. By \prettyref{equ:IP=I} we have $\emptyset P=\emptyset$, for all $P$, which means that the only reasonable thing we can do is to compute $P\emptyset$---but $P\emptyset\stackrel{\prettyref{equ:f(P)}}= f(P)\neq\{r\}$ shows that we cannot add $r$ to the empty program via composition on the left. The intuitive reason is that in order to construct the rule $r$ via composition, we need to be able to add body atoms to facts, which is easily seen to be impossible.
\end{remark}


\subsection{Closures}

The following construction will often be useful when studying decompositions:

\begin{definition} Define the {\em closure} of $P$ with respect to some alphabet $A$ by
\begin{align}\label{equ:cl} 
     c_A(P):=1_A\cup P.
\end{align}
\end{definition}

Roughly, we obtain the closure of a program by adding all possible tautological rules of the form $a\leftarrow a$, for $a\in A$. Of course, the closure operator preserves semantic equivalence, that is,
\begin{align*} 
    P\equiv  c_A(P),\quad\text{for all alphabets $A$.}
\end{align*}

\begin{lemma}\label{lem:PQ} For any programs $P$ and $R$, if $P$ does not depend on $R$, we have
\begin{align}\label{equ:PQ} 
    P(Q\cup R)=PQ.
\end{align}
\end{lemma}
\begin{proof} We compute
\begin{align}\label{equ:PQ1}
    P(Q\cup R)
        \stackrel{\prettyref{equ:P^head_circ_^body_R}}=P^{ h(Q\cup R)}\left({^{ b(P)}}(Q\cup R)\right)
        \stackrel{\prettyref{equ:^IPR}}=P^{ h(Q\cup R)}\left({^{ b(P)}}Q\cup{^{ b(P)}}R\right).
\end{align} Since $P$ does not depend on $R$, we have $ b(P)\cap  h(R)=\emptyset $ and, hence,
\begin{align*} 
    ^{ b(P)}R=\{r\in R\mid  b(P)\models  h(r)\}=\{r\in R\mid  h(r)\in  b(P)\}=\emptyset .
\end{align*} Moreover, we have
\begin{align*} 
    P^{ h(Q\cup R)}
        =P^{ h(Q)\cup  h(R)}
        =\{r\in P\mid  b(r)\subseteq  h(Q)\cup  h(R)\}
        =\{r\in P\mid  b(r)\subseteq  h(Q)\}
        =P^{ h(Q)}.
\end{align*} Hence, \prettyref{equ:PQ1} is equivalent to
\begin{align*} 
    \left(P^{ h(Q)}\right)\left({^{ b(P)}}Q\right)\stackrel{\prettyref{equ:P^head_circ_^body_R}}=PQ.
\end{align*}
\end{proof}

\begin{lemma}\label{lem:P_cup_R} For any programs $P$ and $R$, in case $P$ does not depend on $R$, we have
\begin{align}\label{equ:P_cup_R} 
    P\cup R= c_{ h(R)}(P) c_{ b(P)}(R) \quad\text{and}\quad  c_A(P\cup R)= c_A(P) c_A(R).
\end{align} Moreover, for any alphabets $A$ and $B$, we have
\begin{align*} 
     c_A( c_B(P))= c_{A\cup B}(P).
\end{align*}
\end{lemma}\begin{proof} We compute
\begin{align*} 
     c_{ h(R)}(P) c_{ b(P)}(R)&=(1_{ h(R)}\cup P)(1_{ b(P)}\cup R)\\
    &\stackrel{\prettyref{equ:(P_cup_Q)_circ_R}}=1_{ h(R)}(1_{ b(P)}\cup R)\cup P(1_{ b(P)}\cup R)\\
    &\stackrel{\prettyref{equ:KP_cup_R}}=1_{ h(R)}1_{ b(P)}\cup 1_{ h(R)}R\cup P(1_{ b(P)}\cup R)\\
    &\stackrel{\prettyref{equ:1I_circ_J},\prettyref{equ:^IP}}=1_{ h(R)\cap  b(P)}\cup{^{ h(R)}}R\cup P(1_{ b(P)}\cup R)\\
    &=R\cup P(1_{ b(P)}\cup R)\\
    &\stackrel{\prettyref{equ:PQ}}=R\cup P1_{ b(P)}\\
    &\stackrel{\prettyref{equ:^IP}}=R\cup P^{ b(P)}\\
    &=R\cup P
\end{align*} where the fourth equality follows from $1_{ h(R)\cap  b(P)}=\emptyset $ and $^{ h(R)}R=R$, the fifth equality holds since $P$ does not depend on $R$ (\prettyref{lem:PQ}), and the last equality follows from $P^{ b(P)}=P$.

For the second equality, we compute
\begin{align*} 
     c_A(P) c_A(R)=(1\cup P)(1\cup R)\stackrel{\prettyref{equ:(P_cup_Q)_circ_R}}=1\cup R\cup P(1\cup R)\stackrel{\prettyref{equ:PQ}}=1\cup R\cup P= c_A(P\cup R)
\end{align*} where the third identity follows from the fact that $P$ does not depend on $R$ (\prettyref{lem:PQ}).

Lastly, we compute
\begin{align*} 
     c_A( c_B(P))=1_A\cup 1_B\cup B\stackrel{\prettyref{equ:1^I_cup_J}}=1_{A\cup B}\cup P= c_{A\cup B}(P).
\end{align*}
\end{proof}


\subsection{Acyclic Programs}\label{sec:Acyclic_Programs}

Recall that a program is called {\em acyclic} if it has a level mapping so that the head of each rule is strictly greater than each of its body atom with respect to that mapping. Acyclic programs appear frequently in practice \cite{Apt91}, which is the motivation for studying their decompositions in this section. The main result of this section shows that acyclic programs can be decomposed into single-rule programs (\prettyref{thm:acyclic}).

\begin{example}\label{exa:Ean} We define a family of acyclic programs over $A$, which we call {\em elevator programs}, as follows: given a sequence $(a_1,\ldots,a_n)\in A^n$, $1\leq n\leq |A|$, of distinct atoms, let $E_{(a_1,\ldots,a_n)}$ be the acyclic program
\begin{align*} 
    E_{(a_1,\ldots,a_n)}=\{a_1\}\cup\{a_i\leftarrow a_{i-1}\mid 2\leq i\leq n\}.
\end{align*} So, for instance, $E_{(a,b,c)}$ is the program
\begin{align*} 
    E_{(a,b,c)}= \left\{
    \begin{array}{l}
      a\\
      b\leftarrow a\\
      c\leftarrow b
    \end{array}
    \right\}.
\end{align*} The mapping $\ell$ given by $\ell(a)=1$, $\ell(b)=2$, and $\ell(c)=3$ is a level mapping for $E_{(a,b,c)}$ and we can decompose $E_{(a,b,c)}$ into a product of single-rule programs by
\begin{align*} 
    E_{(a,b,c)}&=(1^{\{b,c\}}\cup\{a\})(1^{\{c\}}\cup \{b\leftarrow a\})(1^{\{a\}}\cup \{c\leftarrow b\})\\
    &= c_{\{b,c\}}(\{a\}) c_{\{c\}}(\{b\leftarrow a\}) c_{\{c\}}(\{c\leftarrow b\}).
\end{align*} 
\end{example}

We now want to generalize the reasoning pattern of \prettyref{exa:Ean} to arbitrary acyclic programs. For this, we will need the following lemma.

\begin{lemma} For any programs $P$ and $R$ and alphabet $B$, if $P$ and $1_B$ do not depend on $R$, we have
\begin{align}\label{equ:c_B} 
     c_B( c_{ h(R)}(P) c_{ b(P)}(R))= c_{B\cup  h(R)}(P) c_{B\cup  b(P)}(R).
\end{align} 
\end{lemma}
\begin{proof} Since $P$ does not depend on $R$, as a consequence of \prettyref{lem:P_cup_R} we have
\begin{align*} 
     c_{ h(R)}(P) c_{ b(P)}=P\cup R.
\end{align*} Hence,
\begin{align*} 
     c_B&( c_{ h(R)}(P) c_{ b(P)}(R))\\
    &= c_B(P\cup R)\\
    &=1_B\cup P\cup R\\
    &\stackrel{\prettyref{equ:P_cup_R}}= c_{ h(R)}(1_B\cup P) c_{ b(1_B\cup P)}(R)\qquad ( h(R)\cap B=\emptyset )\\
    &= c_{B\cup  h(R)}(P) c_{B\cup  b(P)}(R).
\end{align*}
\end{proof}

We further will need the following auxiliary construction. Define, for any linearly ordered rules $r_1<\ldots<r_n$, $n\geq 2$,
\begin{align*} 
    bh_i(\{r_1<\ldots<r_n\}):= b(\{r_1,\ldots,r_{i-1}\})\cup  h(\{r_{i+1},\ldots,r_n\}).
\end{align*}

We are now ready to prove the following decomposition result for acyclic programs.

\begin{theorem}\label{thm:acyclic} We can sequentially decompose any acyclic program $P=\{r_1<_\ell\ldots<_\ell r_n\}$, $n\geq 2$, linearly ordered by a level mapping $\ell$, into single-rule programs as
\begin{align*} 
    P=\prod_{i=1}^n  c_{bh_i(P)}(\{r_i\}).
\end{align*} This decomposition is unique up to reordering of rules within a single level.\footnote{See the construction of the total ordering $<_\ell$ in Section \ref{sec:Syntax}.}
\end{theorem}
\begin{proof} The proof is by induction on the number $n$ of rules in $P$. For the induction hypothesis $n=2$ and $P=\{r_1<_\ell r_2\}$, we proceed as follows. First, we have
\begin{align}\label{equ:acyclic1} 
     c_{bh_1(P)}(\{r_1\}) c_{bh_2(P)}(\{r_2\})&=(1_{ h(\{r_2\})}\cup\{r_1\})(1_{ b(\{r_1\})}\cup\{r_2\})\\
    &=1_{ h(\{r_2\})}(1_{ b(\{r_1\})}\cup\{r_2\})\cup\{r_1\}(1_{ b(\{r_1\})}\cup\{r_2\}).
\end{align} Since Krom programs distribute from the left (\prettyref{thm:Krom}), we can simplify \prettyref{equ:acyclic1}, by applying \prettyref{equ:^IP} and \prettyref{equ:1^I_cup_J}, into
\begin{align}\label{equ:acyclic2} 
    1_{ h(\{r_2\})\cap  b(\{r_1\})}\cup{^{ h(\{r_2\})}}\{r_2\}\cup\{r_1\}^{ b(\{r_1\})}=1_{ h(\{r_2\})\cap  b(\{r_1\})}\cup\{r_1,r_2\}.
\end{align} Now, since $r_1$ does not depend on $r_2$, that is, $ h(r_2)\cap  b(r_1)=\emptyset $, the first term in \prettyref{equ:acyclic2} equals $1_\emptyset =\emptyset $ which implies that \prettyref{equ:acyclic2} is equivalent to $P$ as desired.

For the induction step $P=\{r_1<_\ell\ldots<_\ell r_{n+1}\}$, we proceed as follows. First, by definition of $bh_i$, we have
\begin{align} 
    \prod_{i=1}^{n+1}( c_{bh_i(P)}(\{r_i\}))
        &=\prod_{i=1}^{n+1}(1_{bh_i(P)}\cup\{r_i\})\\
        \label{equ:Ab_1}&=\left[\prod_{i=1}^n(1_{bh_i(\{r_1,\ldots,r_n\})}\cup 1_{ h(r_{n+1})}\cup\{r_i\})\right](1_{ b(\{r_1,\ldots,r_n\})}\cup\{r_{n+1}\}).
\end{align} Second, by idempotency of union we can extract the term $$1_{ h(r_{n+1})}=1_{ h(r_{n+1})}\ldots 1_{ h(r_{n+1})}\quad\text{($n$ times)}$$ occurring in \prettyref{equ:Ab_1} thus obtaining
\begin{align}\label{equ:Ab_2}
    \left[1_{ h(r_{n+1})}\cup\prod_{i=1}^n(1_{bh_i(\{r_1,\ldots,r_n\})}\cup 1_{ h(r_{n+1})}\cup\{r_i\})\right](1_{ b(\{r_1,\ldots,r_n\})}\cup\{r_{n+1}\}).
\end{align} Now, since $r_i$ and $1_{bh_i(\{r_1,\ldots,r_n\})}$ do not depend on $r_{n+1}$, we can simplify \prettyref{equ:Ab_2} further to
\begin{align}\label{equ:Ab_3} 
    \left[1_{ h(r_{n+1})}\cup\prod_{i=1}^n(1_{bh_i(\{r_1,\ldots,r_n\})}\cup\{r_i\})\right](1_{ b(\{r_1,\ldots,r_n\})}\cup\{r_{n+1}\}).
\end{align} By applying the induction hypothesis to \prettyref{equ:Ab_3}, we obtain
\begin{align*} 
    (1_{ h(r_{n+1})}\cup\{r_1,\ldots,r_n\})(1_{ b(\{r_1,\ldots,r_n\})}\cup\{r_{n+1}\})
\end{align*} which, by right-distributivity of composition, is equal to
\begin{align}\label{equ:Ab_4}  
    1_{ h(r_{n+1})}(1_{ b(\{r_1,\ldots,r_n\})}\cup\{r_{n+1}\})\cup\{r_1,\ldots,r_n\}(1_{ b(\{r_1,\ldots,r_n\})}\cup\{r_{n+1}\}).
\end{align} Now again since $r_i$ does not depend on $r_{n+1}$, for all $1\leq i\leq n$, as a consequence of \prettyref{equ:KP_cup_R}, \prettyref{equ:^IP}, \prettyref{equ:1^I_cup_J}, and \prettyref{equ:PQ}, the term in \prettyref{equ:Ab_4} equals
\begin{align}\label{equ:acyclic_final} 
    1_{ h(r_{n+1})\cap  b(\{r_1,\ldots,r_n\})}\cup{^{ h(r_{n+1})}}\{r_{n+1}\}\cup\{r_1,\ldots,r_n\}^{ b(\{r_1,\ldots,r_n\})}.
\end{align} Finally, since $1_{ h(r_{n+1})\cap  b(\{r_1,\ldots,r_n\})}=1_\emptyset =\emptyset $, \prettyref{equ:acyclic_final} equals $\{r_1,\ldots,r_{n+1}\}$, which proves our theorem.
\end{proof}

\subsection{General Decompositions}\label{sec:General_Decompositions}

In the last subsection, we studied the decomposition of acyclic programs and showed that they are assembled of single-rule programs. In this subsection, we study decompositions of arbitrary programs (\prettyref{thm:P_cup_R}) which turns out to be much more difficult given that we now have to deal with circularities among the rules (\prettyref{problem:decomposition}).

We first wish to generalize \prettyref{lem:P_cup_R} to arbitrary propositional logic programs. For this, we define, for every interpretation $I$ and disjoint copy $I'=\{a'\mid a\in I\}$ of $I$, the minimalist program (recall that $A$ is the underlying propositional alphabet)
\begin{align*} 
    [I\leftarrow I']
        &:= c_{A-I}(\{a\leftarrow a'\mid a\in I\})\\
        &\;=\{a\leftarrow a'\mid a\in I\}\cup\{a\leftarrow a\mid a\in A-I\}.
\end{align*} We have $$[A\leftarrow A']=\{a\leftarrow a'\mid a\in A\}$$ and therefore
\begin{align*} 
    P[A\leftarrow A']=\{ h(r)\leftarrow  b(r)'\mid r\in P\},
\end{align*} where $ b(r)'=\{b'\mid b\in  b(r)\}$, and
\begin{align*} 
    [A'\leftarrow A]P=\{ h(r)'\leftarrow  b(r)\mid r\in P\}.
\end{align*} Moreover, we have
\begin{align}\label{equ:A_A'} 
    [A\leftarrow A'][A'\leftarrow A]=1_A.
\end{align}



We are now ready to prove the main result of this subsection which shows that we can represent the union of programs via composition and closures:

\begin{theorem}\label{thm:P_cup_R} For any programs $P$ and $R$, we have 
\begin{align*} 
    P\cup R= c_{ h(R)}(P[A\leftarrow A']) c_{ b(P[A\leftarrow A'])}(R) c_A([A'\leftarrow A]).
\end{align*}
\end{theorem}
\begin{proof} Since $P[A\leftarrow A']$ does not depend on $R$, we have, as a consequence of \prettyref{lem:P_cup_R},
\begin{align*} 
    P[A\leftarrow A']\cup R= c_{ h(R)}(P[A\leftarrow A']) c_{ b(P[A\leftarrow A'])}(R).
\end{align*} Finally, we have
\begin{align*} 
    (P[A\leftarrow A']\cup R) c_A([A'\leftarrow A])&\stackrel{\prettyref{equ:(P_cup_Q)_circ_R}}=P[A\leftarrow A'] c_A([A'\leftarrow A])\cup R\, c_A([A'\leftarrow A])\\
    &=P[A\leftarrow A'](1_A\cup [A'\leftarrow A])\cup R(1_A\cup [A'\leftarrow A])\\
    &=P[A\leftarrow A'][A'\leftarrow A]\cup R\\
    &\stackrel{\prettyref{equ:A_A'}}=P\cup R
\end{align*} where the third identity follows from the fact that $P[A\leftarrow A']$ does not depend on $1_A$ and $R$ does not depend on $[A'\leftarrow A]$ (apply \prettyref{lem:PQ}).
\end{proof}

\begin{example} Let $A=\{a,b,c\}$ and consider the program $P=R\cup\pi_{(b\,c)}$ where
\begin{align*} 
    R= \left\{
    \begin{array}{l}
      a\leftarrow b,c\\
      a\leftarrow a,b\\
      b\leftarrow a
    \end{array}
    \right\} \quad\text{and}\quad\pi_{(b\,c)}= \left\{
    \begin{array}{l}
      b\leftarrow c\\
      c\leftarrow b
    \end{array}
    \right\}.
\end{align*} We wish to decompose $P$ into a product of $R$ and $\pi_{(b\,c)}$ according to \prettyref{thm:P_cup_R}. For this, we first have to replace the body of $R$ with a distinct copy of its atoms, that is, we compute
\begin{align*} 
    R[A\leftarrow A']= \left\{
    \begin{array}{l}
      a\leftarrow b',c'\\
      a\leftarrow a',b'\\
      b\leftarrow a'
    \end{array}
    \right\}.
\end{align*} Notice that $R[A\leftarrow A']$ no longer depends on $\pi_{(b\,c)}$! Next, we compute the composition of
\begin{align*} 
     c_{ h(\pi_{(b\,c)})}(R[A\leftarrow A'])= \left\{
    \begin{array}{l}
      a\leftarrow b',c'\\
      a\leftarrow a',b'\\
      b\leftarrow a'\\
      b\leftarrow b\\
      c\leftarrow c
    \end{array}
    \right\}
\end{align*} and
\begin{align*} 
     c_{ b(R[A\leftarrow A'])}(\pi_{(b\,c)})= \left\{
    \begin{array}{l}
      b\leftarrow c\\
      c\leftarrow b\\
      a'\leftarrow a'\\
      b'\leftarrow b'\\
      c'\leftarrow c'
    \end{array}
    \right\}
\end{align*} by
\begin{align*} 
    R'= c_{ h(\pi_{(b\,c)})}(R[A\leftarrow A']) c_{ b(R[A\leftarrow A'])}(\pi_{(b\,c)})= \left\{
    \begin{array}{l}
      a\leftarrow b',c'\\
      a\leftarrow a',b'\\
      b\leftarrow a'\\
      b\leftarrow c\\
      c\leftarrow b
    \end{array}
    \right\}.
\end{align*} Finally, we need to replace the atoms from $A'$ in the body of $R'$ by atoms from $A$, that is, we compute
\begin{align*} 
    R' c_A([A'\leftarrow A])= \left\{
    \begin{array}{l}
      a\leftarrow b',c'\\
      a\leftarrow a',b'\\
      b\leftarrow a'\\
      b\leftarrow c\\
      c\leftarrow b
    \end{array}
    \right\}\circ \left\{
    \begin{array}{l}
      a\leftarrow a\\
      b\leftarrow b\\
      c\leftarrow c\\
      a'\leftarrow a\\
      b'\leftarrow b\\
      c'\leftarrow c
    \end{array}
    \right\}=P
\end{align*} as expected.
\end{example}

\begin{problem}\label{problem:decomposition} Unfortunately, at the moment we have no such nice decomposition result for arbitrary programs as for acyclic programs (\prettyref{thm:acyclic}), not relying on auxiliary atoms in $A'$, which remains as future work (cf. \prettyref{sec:Future_Work}).
\end{problem}

\section{Algebraic Semantics}\label{sec:Algebraic_Semantics}

Recall that the fixed point semantics of a propositional logic program is defined in terms of least fixed point computations of its associated van Emden-Kowalski immediate consequence operator (\prettyref{prop:LM}). That is, the semantics is defined {\em operationally} outside the syntactic space of the programs. In this section, we reformulate the fixed point semantics of programs in terms of sequential composition by first showing that we can express the van Emden-Kowalski operators via composition (\prettyref{thm:T_P}), and then showing how programs (not operators) can be iterated bottom-up to obtain their least models (\prettyref{thm:LM}). This bridges the conceptual gap between the syntax and semantics of a program.

\subsection{The van Emden-Kowalski Operator}\label{sec:T_P}

The next results show that we can simulate the van Emden-Kowalski operators on a {\em syntactic} level without any explicit reference to operators.

\begin{theorem}\label{thm:T_P} For any program $P$ and interpretation $I$, we have
\begin{align}\label{equ:T_P} 
    T_P(I)=PI.
\end{align} Moreover, we have for any program $R$,
\begin{align}\label{equ:T_PR}
    T_P\cup T_R=T_{P\cup R} \quad\text{and}\quad T_P\circ T_R=T_{PR} \quad\text{and}\quad T_\emptyset =\emptyset  \quad\text{and}\quad T_1=Id_{2^A}.
\end{align} 
\end{theorem}
\begin{proof} 
All of the equations follow immediately from the definition of composition. To illustrate some of the already derived results, we compute
\begin{align*} 
    T_P(I)&= h(P^I)\stackrel{\prettyref{equ:head_P}}=P^IA\stackrel{\prettyref{equ:^IP}}=(P(1^I))A=P((1^I)A)\stackrel{\prettyref{equ:^IP}}=P(^IA)\stackrel{\prettyref{equ:^JI}}=P(I\cap A)=PI
\end{align*} where the last equality follows from $I$ being a subset of $A$. 

\end{proof}


\begin{corollary}\label{cor:assoc_ss} Sequential composition is associative modulo subsumption equivalence, that is,
\begin{align*} 
    P(QR)\equiv_{ss} (PQ)R\quad\text{for any programs $P,Q,R$}.
\end{align*} This implies
\begin{align}\label{equ:PRI} 
    (PR)I=P(RI),\quad\text{for any programs $P,R$ and interpretation $I$}.
\end{align}
\end{corollary}
\begin{proof} Associativity modulo subsumption equivalence is an immediate consequence of \prettyref{thm:T_P} which shows
\begin{align*} 
    T_{P(QR)}=T_{(PQ)R} \quad\Rightarrow\quad P(QR)\equiv_{ss} (PQ)R.
\end{align*} The identity is shown via \prettyref{equ:T_PR} by
\begin{align*} 
    (PR)I=T_{PR}(I)=T_P(T_R(I))=P(RI).
\end{align*}
\end{proof}

As a direct consequence of \prettyref{prop:prefixed} and \prettyref{thm:T_P}, we have the following algebraic characterization of (supported) models:

\begin{corollary} An interpretation $I$ is a model of $P$ iff $PI\subseteq I$, and $I$ is a supported model of $P$ iff $PI=I$.
\end{corollary}

\begin{corollary}\label{cor:I_ideal} The space of all interpretations forms an ideal.
\end{corollary}
\begin{proof} By \prettyref{fact:IP=I}, we know that the space of interpretations forms a right ideal and \prettyref{thm:T_P} implies that it is a left ideal---hence, it forms an ideal.
\end{proof}



\subsection{Least Models}\label{sec:LM}

We interpret propositional logic programs according to their least model semantics and since least models can be constructively computed by bottom-up iterations of the associated van Emden-Kowalski operators (\prettyref{prop:LM}), and since these operators can be represented in terms of composition (\prettyref{thm:T_P}), we can finally reformulate the fixed point semantics of programs algebraically in terms of composition (\prettyref{thm:LM}). 

We make the convention that for any $n\geq 3$,
\begin{align*} 
    P^n=(((PP)P)\ldots P)P\quad (n\text{ times}).
\end{align*}

\begin{definition} Define the unary {\em Kleene star} and {\em plus operations} by
\begin{align}\label{equ:P^ast} 
    P^\ast=\bigcup_{n\geq 0} P^n \quad\text{and}\quad P^+=P^\ast P.
\end{align} Moreover, define the {\em omega operation} by
\begin{align*} 
    P^\omega=P^+\emptyset \stackrel{\prettyref{equ:f(P)}}= f(P^+).
\end{align*}
\end{definition}

Notice that the unions in \prettyref{equ:P^ast} are finite since $P$ is finite. For instance, for any interpretation $I$, we have as a consequence of \prettyref{fact:I_idempotent},
\begin{align}\label{equ:I^ast} 
    I^\ast=1\cup I \quad\text{and}\quad I^+=I \quad\text{and}\quad I^\omega=I.
\end{align} Interestingly enough, we can add the atoms in $I$ to $P$ via
\begin{align}\label{equ:P_cup_I} 
    P\cup I\stackrel{\prettyref{equ:IP=I}}=P\cup IP\stackrel{\prettyref{equ:(P_cup_Q)_circ_R}}=(1\cup I)P\stackrel{\prettyref{equ:I^ast}}=I^\ast P.
\end{align} Hence, as a consequence of \prettyref{equ:P} and \prettyref{equ:P_cup_I}, we can decompose $P$ as
\begin{align*} 
    P= f(P)^\ast  p(P),
\end{align*} which, roughly, says that we can sequentially separate the facts from the proper rules in $P$.

We are now ready to characterize the least model of a program via composition as follows.

\begin{theorem}\label{thm:LM} For any program $P$, we have
\begin{align*} 
    LM(P)=P^\omega.
\end{align*}
\end{theorem}
\begin{proof} A direct consequence of \prettyref{prop:LM} and \prettyref{thm:T_P}.
\end{proof}

\begin{corollary}\label{cor:P^omega=R^omega} Two programs $P$ and $R$ are equivalent iff $P^\omega=R^\omega$.
\end{corollary}

To summarize, we have thus shown that we can use the syntactic sequential composition operation to capture the least model or fixed point semantics of a program algebraically {\em within} the space of programs giving rise to the algebraic characterization of the semantics of a program in \prettyref{thm:LM} and the semantic equivalence of programs in \prettyref{cor:P^omega=R^omega}.

\section{Related Work}\label{sec:RW}

\citeA{OKeefe85} is probably the first to study the composition of logic programs. He first defines the denotation of $P$ by\footnote{For an operator $T$, we define $T^\ast$ in the same way as $P^\ast$ in \prettyref{sec:LM} as $T^\ast=Id\cup T\cup T^2\cup\ldots$; notice that \citeA{OKeefe85} uses the notation $T^\omega$ instead of $T^\ast$ which we refuse to use here to remain consistent with the use of the symbols in \prettyref{sec:LM}.} $\llbracket P\rrbracket=(T_P\cup Id)^\ast$ and then defines the semantic composition of $P$ and $R$ so that $\llbracket P\circ R\rrbracket=\llbracket P\rrbracket\circ\llbracket R\rrbracket$ \cite<cf.>[p. 450]{Bugliesi94}. That is, he defines the composition of two programs on a {\em semantic} level in terms of the compositions of their immediate consequence operators which is straightforward given that we already know how to compose functions; in contrast, in our approach of this paper we define the sequential composition of two programs explicitly on a {\em syntactic} level which allows us to distinguish algebraically between subsumption equivalent programs (at the cost of losing associativity; see \prettyref{exa:non-associativity} and \prettyref{cor:assoc_ss}). 

It is important to emphasize that since we can represent the immediate consequence operator in terms of sequential composition (\prettyref{thm:T_P}), O'Keefe's definition can be resembled in terms of sequential composition as well via\footnote{Notice that the identity function $Id$ corresponds to the unit program $1$ since $Id(I)=1\circ I$ holds for all interpretations $I$.}
\begin{align*} 
    \llbracket P\circ R\rrbracket(I)
        &=(\llbracket P\rrbracket\circ\llbracket R\rrbracket)(I)\\
        &=((T_P\cup Id)^\ast\circ (T_R\cup Id)^\ast)(I)\\
        &=(T_P\cup Id)^\ast((T_R\cup Id)^\ast(I))\\
        &\stackrel{\prettyref{equ:T_P}}=(P\cup 1)^\ast\circ ((P\cup 1)^\ast\circ I)\\
        &\stackrel{\prettyref{equ:cl}}=  c(P)^\ast\circ (  c(R)^\ast\circ I)\quad\text{for all interpretations $I$.}
\end{align*} 

\citeA{Bugliesi94} study {\em modular logic programming} by defining algebraic operations on programs similar to \citeA{OKeefe85}. More formally, O'Keefe's composition is captured in \citeS{Bugliesi94} framework as $\llbracket P\circ R\rrbracket=\llbracket(P\cup R^\omega)^\omega\rrbracket$, where $\llbracket P^\omega\rrbracket=T_P^\infty$ in our notation \cite<cf.>[p. 450]{Bugliesi94}. This shows that we can represent \citeS{Bugliesi94} algebraic operations in terms of sequential composition as well. The work of \citeA{Brogi92} is similar in spirit to the aforementioned works. Formally, \citeA{Brogi92} first define the admissible Herbrand model of $P$ denoted (again) by $\llbracket P\rrbracket$ and then define a {\em semantic} composition operation $\bullet$ which allows one to compute the semantics of the union of two programs in the sense that $\llbracket P\cup R\rrbracket=\llbracket P\rrbracket\bullet\llbracket R\rrbracket$ \cite<cf.>[Proposition 3.4]{Brogi92}. 

\citeA{Dong90,Dong95} study the composition and decomposition of datalog program {\em mappings}---where datalog programs are function-free logic programs used mainly in database theory \cite<cf.>{Ceri89,Ceri90}---for query optimization in the context of databases. \citeA{Plambeck90,Plambeck90a} employs tools from semigroup theory for query optimization. His work is similar in spirit to the work of \citeA{Ioannidis91} who map Horn rules to an relational algebra operator, where the set of all such operators forms a semiring, and answer queries via solving linear equations. The authors then study decompositions of relational algebra operators for optimization purposes. To the best of our knowledge, these are the only works utilizing semigroups and semirings in logic programming.\footnote{There are semiring-based works on {\em weighted} logic programs, where semirings are used to add probabilities and other quantitative information to rules and programs \cite<e.g.>{Bistarelli01}.}


\citeA{Dix01} consider the following {\em principle of partial evaluation} as a semantic preserving transformation of a given program $P$. Let $r=a_0\leftarrow a_1,\ldots,a_k\in P$ and choose some fixed body atom $a$ among $a_1,\ldots,a_k$; without loss of generality, we can assume $a=a_1$. If all the rules in $P$ with $a_1$ in the head are given by
\begin{align}\label{equ:a_1} 
    a_1\leftarrow B_1 \quad\ldots\quad a_1\leftarrow B_n,
\end{align} then we can replace $r$ by the $n$ rules
\begin{align*} 
    a_0\leftarrow B_1\cup\{a_2,\ldots,a_k\} \quad\ldots\quad a_0\leftarrow B_n\cup\{a_2,\ldots,a_k\}.
\end{align*} There is some similarity between this transformation and the computation of $P^2=P\circ P$. Formally, if we add to $P$ the tautological rules
\begin{align*} 
    a_2\leftarrow a_2 \quad\ldots\quad a_k\leftarrow a_k,
\end{align*} then $P^2$ contains all the rules in \prettyref{equ:a_1}; however, $P^2$ may contain many more rules with $a_1$ in its head depending on the rules in $P$ with $a_2,\ldots,a_k$ as head atoms. So the main difference between the above principle and the sequential composition operation of this paper is that the former operates on {\em single} programs and on {\em single} body atoms whereas the latter is defined as a binary operation between {\em two} programs and operates on {\em all} body atoms at once.

The sequential composition and decomposition of (first-order) logic programs is used in the context of logic program synthesis via analogical proportions in \citeA{Antic23-23} and for defining a qualitative notion of syntactic program similarity, which allows one to answer queries across different domains, in \citeA{Antic21-4}. \citeA{Antic23-1} provides algebraic characterizations of least model and uniform equivalence of propositional Krom logic programs in terms of sequential composition. 

From a purely mathematical point of view, sequential composition of propositional logic programs as introduced in this paper is very similar to composition in multicategories \cite<cf.>[§2.1]{Leinster04}.

\section{Future Work}\label{sec:Future_Work}



In the future, we plan to extend the constructions and results of this paper to wider classes of logical formalisms, most importantly to first-order and higher-order logic programs \cite{Lloyd87,Apt90}, disjunctive programs \cite{Eiter97}, and non-monotonic logic programs under the stable model or answer set semantics \cite{Gelfond91} \cite<see>{Baral03,Brewka11,Eiter09,Lifschitz19}. The first task is non-trivial as function symbols require the use of most general unifiers in the definition of composition and give rise to infinite algebras, whereas the non-monotonic case is more difficult to handle algebraically due to negation as failure \cite{Clark78} occurring in rule bodies (and heads). In \citeA{Antic21-2}, the author studies the non-monotonic case in detail by lifting the concepts and results of this paper from the Horn to the general case containing negation as failure. Some of the results of this paper are subsumed by more general results in \citeA{Antic21-2}. However, lifting the results on decompositions of programs of \prettyref{sec:Decomposition} to programs with negation turns out to be highly non-trivial and therefore remains an open problem. Even more problematic, disjunctive rules yield non-deterministic behavior which is more difficult to handle algebraically. Nonetheless, we expect interesting results in all of the aforementioned cases to follow.


Another major line of research is to study sequential {\em decompositions} of programs as was initiated in \prettyref{problem:decomposition}. Specifically, we wish to compute decompositions of arbitrary propositional logic programs (and extensions thereof) into simpler programs in the vein of \prettyref{thm:acyclic}, where we expect permutation programs (\prettyref{sec:Permutations}) to play a fundamental role in such decompositions. For this, it will be necessary to resolve the issue of defining the notion of a ``prime'' or indecomposable program. From a practical point of view, a mathematically satisfactory theory of program decompositions is relevant to modular logic programming and optimization of reasoning.



\section{Conclusion}

This paper studied the (sequential) composition of propositional logic programs. We showed in our main structural result (\prettyref{thm:Horn}) that the space of all such programs forms a finite unital magma with respect to composition ordered by set inclusion, which distributes from the right over union. 
Moreover, we showed that the restricted class of propositional Krom programs is distributive and therefore its proper instance forms an idempotent semiring (\prettyref{thm:Krom}). From a logical point of view, we obtained algebraic tools for reasoning about propositional logic programs. Algebraically, we obtained a correspondence between propositional logic programs and finite unital magmas and other algebraic structures, which hopefully enables us in the future to transfer concepts from the algebraic literature to the logical setting. In a broader sense, this paper is a first step towards an algebra of logic programs and we expect interesting concepts and results to follow.

\if\isdraft0
\section*{Acknowledgments}

We would like to thank the reviewers for their thoughtful and constructive comments, and for their helpful suggestions to improve the presentation of the article.


\section*{Conflict of interest}

The authors declare that they have no conflict of interest.

\section*{Data availability statement}

The manuscript has no data associated.
\fi

\if\isdraft1\newpage\fi
\bibliographystyle{theapa}
\bibliography{/Users/christianantic/Bibdesk/Bibliography,/Users/christianantic/Bibdesk/Preprints,/Users/christianantic/Bibdesk/Publications,/Users/christianantic/Bibdesk/Unpublished}
\end{document}